\documentclass[11pt]{amsart}
\usepackage[utf8]{inputenc}
\usepackage[margin=1in]{geometry}
\usepackage{hyperref}
\usepackage{amsfonts}
\usepackage{amsmath}
\usepackage{amsthm}
\usepackage{amssymb}
\usepackage{mathrsfs}
\usepackage{tikz}

\numberwithin{equation}{section}

\newcommand{\R}{\mathbb{R}}
\newcommand{\N}{\mathbb{N}}

\newcommand{\C}{\mathbb{C}}

\newcommand{\wt}[1]{\widetilde{#1}}

\newcommand{\dom}{\operatorname{dom}}
\newcommand{\ran}{\operatorname{ran}}

\theoremstyle{definition}
\newtheorem{defn}{Definition}
\newtheorem{lem}{Lemma}
\newtheorem*{rem}{Remark}
\newtheorem{thm}{Theorem}

\newtheorem{exmp}{Example}

\title[Coupled Supersymmetry]{Coupled Supersymmetry and \\ 
Ladder Structures Beyond the Harmonic Oscillator}
\author{Cameron L. Williams}
\address{
Department of Mathematics \\ 
University of Houston \\ 
Houston, TX 77204-3008, USA}
\author{Nikhil N. Pandya}
\address{
Departments of Physics and Mathematics \\
University of Houston \\ 
Houston, TX 77204-5005, USA}
\author{Bernhard G. Bodmann}
\address{
Department of Mathematics \\ 
University of Houston \\ 
Houston, TX 77204-3008, USA}
\author{Donald J. Kouri}
\address{
Departments of Physics, Mathematics and Mechanical Engineering \\
University of Houston \\ 
Houston, TX 77204-5005, USA}
\date{}

\begin{document}

\begin{abstract}
The development of supersymmetric (SUSY) quantum mechanics has shown that some of the insights based on the algebraic properties of ladder operators related to the quantum mechanical harmonic oscillator carry over to the study of more general systems. At this level of generality, pairs of eigenfunctions of so-called partner Hamiltonians are transformed into each other, but the entire spectrum of any one of them cannot be deduced from this intertwining relationship in general---except in special cases. In this paper, we present a more general structure that provides all eigenvalues for a class of Hamiltonians that do not factor into a pair of operators satisfying canonical commutation relations. Instead of a pair of partner Hamiltonians, we consider two pairs that differ by an overall shift in their spectrum. This is called coupled supersymmetry. In that case, we also develop coherent states and present some uncertainty principles which generalize the Heisenberg uncertainty principle. Coupled SUSY is explicitly realized by an infinite family of differential operators which admit orthonormal bases of eigenfunctions of generalized harmonic oscillators.
\end{abstract}

\maketitle

\section{Introduction}

The quantum mechanical harmonic oscillator (QMHO) is an important part of quantum theory. It has close connections to classical mechanics through an associated family of coherent states that minimize the Heisenberg uncertainty principle for the position and momentum operators. Often, the harmonic oscillator is used as an approximation to describe quantum systems that oscillate about an equilibrium position. It is also fundamental in quantum field theory and related areas like BCS superconductivity in solid state theory \cite{haken}. In the position representation, measuring energy in units of $\hbar \omega$, where $2\pi \hbar$ is Planck's constant and $\omega$ is the angular frequency of the oscillator, the Hamiltonian of the QMHO is the Schr{\"o}dinger operator
%
$
    \mathcal{H}_{\text{HO}} = -\frac{1}{2}\frac{d^2}{dx^2} + \frac{1}{2}x^2 .
$
The domain of this operator is most easily described in terms of the sequence of Hermite functions $\{h_k\}_{k=0}^\infty$ which form an orthonormal basis of eigenfunctions for the Hilbert space $L^2(\mathbb R)$. Each function $\psi$ in the domain of $ \mathcal{H}_{\text{HO}}$ is given by a series expansion $\psi= \sum_{k=0}^\infty c_k h_k$ whose sequence of coefficients $\{c_k\}_{k=0}^\infty$ is such that the norm obtained after applying $ \mathcal{H}_{\text{HO}}$ is finite, $\sum_{k=0}^\infty k^2 |c_k|^2 < \infty$.

The eigenfunctions of the oscillator Hamiltonian are also eigenfunctions of the Fourier transform, which can be obtained from the time evolution generated by $\mathcal{H}_{\text{HO}}$. In an earlier paper \cite{williams}, three of the authors investigated generalized Fourier transforms that emerged from an axiomatic framework. These were identified as the unitary transforms that diagonalized certain singular Laplacian operators $\Delta_n = -\frac{d}{dx}\frac{1}{x^{2n-2}}\frac{d}{dx}$. The motivation for the present paper came from the question whether these transforms are related to a type of oscillator Hamiltonian. This is indeed the case. If $a_n = \frac{1}{\sqrt{2}} \left(\frac{1}{x^{n-1}} \frac{d}{dx} + x^n\right)$, then $\mathcal H = a_n^* a_n$ has the set of eigenvalues $\{2kn, 2kn+2n-1\}_{k=0}^\infty$ and the corresponding eigenfunctions are complete in $L^2(\R)$. The present paper derives the spectrum of this operator with an algebraic method that extends supersymmetric (SUSY) quantum mechanics to a setting in which \emph{ two } pairs of partner Hamiltonians are not related, except in the case of $n=1$. These generalized oscillator Hamiltonians include a generalized Laplacian which was the study of anomalous and normal diffusion \cite{diffusion} as well as a potential of the form $x^{2n}$ which can be realized as the approximate potential in potential wells. 

The theory of supersymmetric quantum mechanics relates eigenfunctions and the corresponding eigenvalues between two partner Hamiltonians by an intertwining relationship with so-called charge operators \cite{BDN92}. In general, this allows the transfer of knowledge of eigenfunctions from one Hamitonian to its partner, but does not characterize the set of eigenvalues as a whole.
In the case of the harmonic oscillator, the partner Hamiltonians are related by a shift in the spectrum, which allows the retrieval of all eigenvalues. The present paper shows algebraic conditions which provide the full benefit of the ladder structure for a class of systems that we call coupled SUSY which acts as a new, elementary generalization of the harmonic oscillator. Generalizing the harmonic oscillator algebra in various ways is not a new concept. Many such generalizations involve $q$-deformations \cite{bonatsos, Brzezinski, Kosinski1997, Kwek}, commutators involving powers of operators or deformed commutators \cite{Kibler, pujari, SADEGHNEZHAD2016555} that are of interest in quantum gravity, work with structure functions \cite{kyemba}, or involve general energy eigenvalue structures \cite{zelaya}.

The remainder of this paper is organized as follows. In Section 2, we develop the {coupled SUSY} structure which expands the relationship among the QMHO, SUSY, and the corresponding coupled SUSYs. In Section 3, we establish eigenvalues and eigenfunctions of the corresponding coupled SUSY. In Section 4, we establish the connections between coupled SUSY and other oscillator systems, namely Schwinger's approach to the two-particle QMHO system, and the $\mathfrak{su}(1,1)$ approach to the QMHO. In Section 5, we develop the coherent states for coupled SUSY. Due to the algebraic structure of coupled SUSY, a more complex coherent state structure exists than what has been found previously for either SUSY or the QMHO. In Section 6, we derive some uncertainty principles associated to coupled SUSY via ladder operators, generalizing the traditional Heisenberg uncertainty principle. In doing so, we discover new quantum operators corresponding to the classical Lagrangian and dilation. In Section 7, we show that harmonic oscillators can be realized as special classes of coupled SUSYs, suggesting that coupled SUSYs may have utitlity in quantum field theory and elementary particles.

\section{Coupled SUSY and the \texorpdfstring{$\mathfrak{su}(1,1)$}{TEXT} Lie algebra}

We first review the algebraic derivation of the eigenvalues of the harmonic oscillator. Defining the operator $a = \frac{1}{\sqrt{2}} \left(\frac{d}{dx}+x\right)$ and its adjoint $a^* = \frac{1}{\sqrt{2}} \left(-\frac{d}{dx}+x\right)$ on a sufficiently nice subspace of $L^2(\R)$, the Hamiltonian is expressed as
\begin{equation}
    \mathcal{H}_{\text{HO}} = a^*a+\frac{1}{2} = aa^*-\frac{1}{2}.
\end{equation}
 
 The commutation relation  $ [a,a^*] \equiv aa^*-a^*a = 1$  shows that the triple $\{a,a^*,1\}$ generates a Lie algebra. From this it follows that $[\mathcal{H}_{\text{HO}},a^*] = a^*$ and $[\mathcal{H}_{\text{HO}},a] = -a$. Therefore $a^*$ and $a$ are ladder operators for $\mathcal{H}_{\text{HO}}$, i.e. if $\psi$ is an eigenfunction of $\mathcal{H}_{\text{HO}}$, then $a^* \psi$ and $a\psi$ are either annihilated or are also eigenfunctions of $\mathcal{H}_{\text{HO}}$ with a different eigenvalue. We call $a$ the lowering operator for $\mathcal{H}_{\text{HO}}$ and $a^*$  the raising operator for $\mathcal{H}_{\text{HO}}$---note that $a^*$ does not annihilate any state. 

It is straightforward to verify that $h_0(x) = \frac{1}{\pi^{1/4}}\exp(-x^2/2)$ provides a normalized eigenfunction of $a^*a$. Moreover, the existence of ladder operators guarantees that for $m\in\N_0$, $(a^*)^m h_0$ is also an eigenfunction of $a^*a$. Normalizing these eigenfunctions yields the orthonormal basis of $L^2(\mathbb R)$ given by the Hermite functions, so the entire spectrum is characterized in terms of the ladder structure \cite[p. 145]{galindo}.

The theory of supersymmetric quantum mechanics (SUSY) generated insights from the factorization method for  more general quantum mechanical systems   \cite{cooper}. SUSY has found its way into many applications, including differential geometry \cite{Witten82}, the study of coherent states \cite{kouri2,david} and BCS theory \cite{yanagisawa}. In SUSY, one assumes that the ground state energy of the quantum mechanical system is zero since ultimately one only measures energy changes, not absolute energy values, and seeks to write a Hamiltonian as
\begin{equation}
    \mathcal{H}_1 = -\frac{1}{2}\frac{d^2}{dx^2} + V(x) = \frac{1}{2}\left(-\frac{d}{dx}+W_1(x)\right) \left(\frac{d}{dx}+W_1(x)\right).
\end{equation}

\noindent This leads to the Riccati equation $V(x) = \frac{1}{2} W_1(x)^2-\frac{1}{2}W_1'(x).$ In general such differential equations have no closed form solution. If however the (nodeless) ground state wavefunction, $\psi$, is known for the Hamiltonian, one may define $W_1$ by
\begin{equation}
    W_1(x) = -\frac{d}{dx}\log\psi(x).
\end{equation}

Such a choice of $W_1$ solves the Riccati equation. One then defines the so-called ``charge'' operators
\begin{align}
    \mathcal{Q}_1 &= \frac{1}{\sqrt{2}}\left(\frac{d}{dx} + W_1(x)\right), \\
    \mathcal{Q}_1^* &= \frac{1}{\sqrt{2}} \left(-\frac{d}{dx} + W_1(x)\right).
\end{align}

In the QMHO case, $W_1(x) = x$ and $\mathcal{Q}$ and $\mathcal{Q}^*$ are the usual raising and lowering operators. Thus SUSY leads to a theory which somewhat mirrors that of the QMHO: the ground states are Gaussian-like with an integral of $-W_1$ in the exponential, there are associated uncertainty principles which have the ground states as minimizers \cite{kouri2}, and there is a ladder structure of sorts.

The SUSY ladder structure is generally more complicated than that of the traditional QMHO, because $[\mathcal{Q}_1,\mathcal{Q}_1^*] = 1$ if and only if $W_1(x) = x+\alpha$, where $\alpha\in \mathbb R$ is a constant that shifts the position. For general $W_1$, one instead constructs a secondary (partner) Hamiltonian $\mathcal{H}_2$ by 
\begin{equation}
    \mathcal{H}_2 = \mathcal{Q}_1 \mathcal{Q}_1^* = -\frac{1}{2}\frac{d^2}{dx^2} + \frac{1}{2} W_1(x)^2 + \frac{1}{2} W_1'(x).
\end{equation}
For such a choice of $\mathcal{H}_2$, it is clear that $\mathcal{Q}_1 \mathcal{H}_1 = \mathcal{H}_2 \mathcal{Q}_1$ (and similarly $\mathcal{H}_1 \mathcal{Q}_1^* = \mathcal{Q}_1^* \mathcal{H}_2$). Thus $\mathcal{Q}_1$ and $\mathcal{Q}_1^*$ act as intertwining operators for $\mathcal{H}_1$ and $\mathcal{H}_2$ so that an eigenfunction of one Hamiltonian can be converted to an eigenfunction of the other Hamiltonian with identical eigenvalue via these operators (except perhaps for the $0$ eigenvalue).

If $\phi$ is a nodeless ground state eigenfunction of $\mathcal{H}_2$, then one can define a $W_2$ much in the same way as $W_1$ to generate new charge operators $\mathcal{Q}_2$ and $\mathcal{Q}_2^*$ and so on to generate a hierarchy of Hamiltonians that are intertwined via a sequence of charge operators.

Since the charge operators convert between eigenfunctions of the Hamiltonians, knowing the ground state eigenfunctions of the hierarchy of all of the Hamiltonians gives all of the excited states of the original Hamiltonian $\mathcal{H}_1$. However, one may not have exact solutions for the ground states of the various SUSY partners. Approximate eigenvalues can then be obtained via variational techniques \cite{kouri1}. In the exceptional case of the QMHO, all excited states are determined only using $\mathcal{Q}_1$ and $\mathcal{Q}_1^*$. It is therefore of interest to find a more general QMHO-like structure which carries the full benefit of a true ladder structure.

We saw that the QMHO has two separate factorizations: $\mathcal{H}_{\text{HO}} = a^*a + \frac{1}{2}$ and $\mathcal{H}_{\text{HO}} = aa^*-\frac{1}{2}$. Defining then the SUSY Hamiltonian $\mathcal{H}_1 = a^*a$, its SUSY partner Hamiltonian is $\mathcal{H}_2 = aa^*$. From the two factorizations of the QMHO, it is clear that $\mathcal{H}_1$ and $\mathcal{H}_2$ each have two distinct factorizations.

That $\mathcal{H}_2 = \mathcal{H}_1+1$ is a restatement of the commutation relation for $a$ and $a^*$---which is equivalent to the canonical commutation relation. This cannot serve as a point of generalization as the canonical commutation relation is too rigid \cite[p. 274]{reed2}. Instead, we use the property that the QMHO and its partner Hamiltonian both have two distinct factorizations to develop our theory and this leads into our first definition.

\begin{defn}
Let $\mathfrak{H}_1$ and $\mathfrak{H}_2$ be Hilbert spaces, $a:\mathfrak{H}_1\to\mathfrak{H}_2$ and $b:\mathfrak{H}_2\to\mathfrak{H}_1$ be closed, densely defined operators, $a^*$ and $b^*$ be their adjoints, and $\gamma,\delta\in\R\setminus\{0\}$ with $\gamma<\delta$. Furthermore, suppose that $\dom a = \dom b^*$, $\dom a^* = \dom b$, $\ran a\subseteq \dom a^*$ and vice versa so that the products $a^*a$ and $aa^*$ are well-defined, and similarly for $b$ and $b^*$. The ordered quadruplet $\{a,b,\gamma,\delta\}$ defines a \emph{coupled supersymmetry} (\emph{coupled SUSY}) if it satisfies
\begin{align}
    a^*a &= bb^* + \gamma, \label{eq:susy_1} \\
    aa^* &= b^*b + \delta. \label{eq:susy_2}
\end{align}

\noindent The terms $a^*a$, $aa^*$, $b^*b$ and $bb^*$ will be referred to as Hamiltonians throughout the paper.
\end{defn}

It is easily seen that the coupled SUSY conditions are equivalent to the system $[a^*,a] = -[b^*,b]+\gamma-\delta$ and $\{a^*,a\} = \{b^*,b\} + \gamma+\delta$, where $[A,B] = AB-BA$ and $\{A,B\} = AB+BA$, which are to be understood on the intersection of their domains of definition. Note that \eqref{eq:susy_1} and \eqref{eq:susy_2} do not imply each other, as evidenced by taking $a = \frac{1}{\sqrt{2}}\left(\frac{d}{dx}+x\right)$ and $b = \frac{1}{\sqrt{2}}\left(\frac{d}{dx}+x\right) e^{ix}$ on $C^{\infty}_c(\R)$. In this case, \eqref{eq:susy_1} holds whereas \eqref{eq:susy_2} does not. Placing the exponential on the right in the definition of $b$ shows that the reverse implication does not hold either.

The condition that $\gamma \neq \delta$ will play a crucial role in what follows. We assumed without much loss of generality that $\gamma < \delta$; otherwise, if $\gamma > \delta$, we may simply switch the roles of $a$ and $a^*$ as well as $b$ and $b^*$ and the conditions for Definition 1 would still hold.

\begin{exmp}
Definition 1 includes the QMHO by letting $a = \frac{1}{\sqrt{2}} \left(\frac{d}{dx}+x\right)$, $b = \frac{1}{\sqrt{2}} \left(-\frac{d}{dx}+x\right)$, $-\gamma = 1 = \delta$. There exists an infinite family of examples solving the coupled SUSY equations. A straightforward calculation shows that, for $n\in\N$, the operators
\begin{align}
    a_n &= \frac{1}{\sqrt{2}} \left(\frac{1}{x^{n-1}} \frac{d}{dx} + x^n\right) \\
    b_n &= \frac{1}{\sqrt{2}} \left(-\frac{d}{dx}\frac{1}{x^{n-1}} + x^n\right)
\end{align}

\noindent taken with their adjoints on $L^2(\R)$ also define a coupled SUSY when restricted to an appropriate subspace where $\gamma = -1$ and $\delta = 2n-1$. This family of examples is closely related to standard SUSY with the anharmonic superpotential $W_1(x) = x^{2n-1}$ where the charge operator $\frac{d}{dx}+x^{2n-1}$ has been multiplied on the left by $\frac{1}{x^{n-1}}$. While a change of variable relates these operators to the traditional QMHO ladder operators, their respective adjoints have a much different form, creating an altogether new system.
\end{exmp}

Coupled SUSY, as the name suggests, automatically comes with a supersymmetric structure by noting that $\mathcal{H}_1 = a^*a$ and $\mathcal{H}_2 = aa^*$ come equipped with the usual SUSY ladder structure. In SUSY, one considers broken and unbroken (or exact) systems. SUSY is unbroken if at least one of $\mathcal{Q}_1$ and $\mathcal{Q}_1^*$ in the factorization $\mathcal{H}_1 = \mathcal{Q}_1^*\mathcal{Q}_1$ annihilates a state and is broken if neither annihilate a state \cite{cooper}. Unbroken SUSY is the primary focus in the study of supersymmetric quantum mechanics and is the focus of much of the remainder of this paper. Partially broken coupled SUSY will be investigated in Section 7. To this end, we make the following definition.

\begin{defn}
Let $\{a,b,\gamma,\delta\}$ form a coupled SUSY. We say that it is an \emph{unbroken} coupled SUSY if $a$ annihilates a state and $b$ annihilates a state but $a^*$ and $b^*$ do not annihilate any states. Otherwise we say that it is a \emph{broken} coupled SUSY. A broken coupled SUSY in which exactly one of $a$ and $b$ (or $a^*$ and $b^*$) annihilates a state is referred to as a \emph{partially broken} coupled SUSY.
\end{defn}

\begin{rem}
If \eqref{eq:susy_1} and \eqref{eq:susy_2} hold but instead $a^*$ and $b^*$ annihilate states, we may reverse the roles of $a$ and $a^*$ (and $b$ and $b^*$) to obtain an unbroken coupled SUSY. Moreover, in the unbroken case, $\gamma < 0 < \delta$. This inequality plays an important role in the proceeding analysis.
\end{rem}

An example of a partially broken SUSY is given by taking $a = \frac{1}{x^2}\frac{d}{dx}x^2 + x$ and $b = x^2\frac{d}{dx}\frac{1}{x^2} + x$ on sufficiently nice functions in $L^2(\R)$ and their (formal) adjoints. The SUSY generated by $a^*a$ and $aa^*$ is unbroken since $a$ annihilates a state in $L^2(\R)$, whereas the SUSY generated by $b^*b$ and $bb^*$ is broken as neither $b$ nor $b^*$ annihilate states in $L^2(\R)$.

In our first theorem, we show that a ladder structure exists for coupled SUSY that does not exist for general SUSYs.

\begin{thm}
 If $\{a,b,\gamma,\delta\}$ form a coupled SUSY, then $a^*b^*$ and $ba$ act as ladder operators for $a^*a$ (and $b^*b$) while $b^*a^*$ and $ab$ act as ladder operators for $aa^*$ (and $bb^*$). Moreover the triples $\{a^*a - \frac{\gamma}{2}, a^*b^*, ba\}$ and $\{aa^*-\frac{\delta}{2},b^*a^*,ab\}$ generate Lie algebras isomorphic to $\mathfrak{su}(1,1)$.
\end{thm}

\begin{proof}
To prove this, we proceed in much the same way as in the standard QMHO by considering the commutator of $a^*a$ with $a^*b^*$, likewise $a^*b^*$ with $ba$. The other cases with $aa^*$, $b^*a^*$ and $ab$ follow the same logic and so they are omitted for the sake of brevity.
\begin{align}
    [a^*a,a^*b^*] &= a^*aa^*b^* - a^*b^* a^*a \\
    &= a^*aa^*b^* - a^*b^*(bb^*+\gamma) \\
    &= a^*(aa^*-bb^*)b^* - \gamma a^*b^* \\
    &= (\delta-\gamma)a^*b^*.
\end{align}

Similar reasoning shows that $[a^*a,ba] = -(\delta-\gamma)ba$. Since $\gamma<\delta$, $a^*b^*$ is a raising operator for $a^*a$ and $ba$ is a lowering operator for $a^*a$. To show that these generate a Lie algebra, we  inspect the commutator of $a^*b^*$ and $ba$:
\begin{align}
    [a^*b^*,ba] &= a^*b^*ba - baa^*b^* \\
    &= a^*(aa^*-\delta)a - b(b^*b+\delta)b^* \\
    &= (a^*a)^2 - (bb^*)^2 - \delta(a^*a + bb^*) \\ 
    &= -2(\delta-\gamma)\left(a^*a - \frac{\gamma}{2}\right).
\end{align}

Thus the triple generates a Lie algebra as it is closed under commutation. After adding a multiple of the identity to $a^* a$ and rescaling,
\begin{equation}
   \mathcal K_+ = \frac{1}{\delta-\gamma} a^* b^*, \quad \mathcal K_- = \frac{1}{\delta-\gamma} ba, \quad \mathcal K_0 = \frac{1}{\delta-\gamma} \left(a^* a - \frac \gamma 2 \right)
\end{equation}

\noindent are seen to verify
\begin{equation}
    [\mathcal{K}_0, \mathcal{K}_{\pm}] = \pm \mathcal{K}_{\pm}, \quad [\mathcal{K}_+, \mathcal{K}_-] = -2\mathcal{K}_0
\end{equation}

\noindent the commutation relations of $\mathfrak{su}(1,1)$ \cite{perelomov}, hence the Lie algebra associated to coupled SUSY is found to be isomorphic to the $\mathfrak{su}(1,1)$ Lie algebra.
\end{proof}

We note that the discrete series of $\mathfrak{su}(1,1)$ representations correspond to a choice of $\frac{\gamma}{\delta-\gamma} \in \mathbb Z$, $\frac{\gamma}{\delta-\gamma} \le -2$ \cite{perelomov}.

As with any mathematical structure, it is of interest to ask if there are systems for which the commutation relations in Theorem 1 hold that do not correspond to a coupled SUSY---be it a broken or unbroken coupled SUSY. In the next theorem, we prove that this is not the case under the modest assumption that $\ker a^* = \{0\} = \ker b^*$.

\begin{thm}
If $a,a^*,b,b^*$ are operators satisfying $\ker a^* = \{0\} = \ker b^*$ and
\begin{align}
    [a^*a,a^*b^*] &= \lambda a^*b^*, \\
    [b^*b,b^*a^*] &= \lambda' b^*a^*, \\
    [a^*b^*,ba] &= \mu a^*a + \nu, \\
    [b^*a^*,ab] &= \mu'aa^* + \nu',
\end{align}
where $\lambda,\lambda',\nu,\nu'\in\R$, $\lambda \ne 0$, then $bb^* = \alpha a^*a + \gamma$ and $b^*b = \beta aa^* + \delta$ for some $\alpha,\beta,\gamma,\delta\in\R$.
\end{thm}

\begin{proof}
Suppose that $a^*a = bb^* + S$ and $aa^* = b^*b + T$ for some as-of-yet undetermined operators $S$ and $T$. Inspecting commutation relations, we have
\begin{align}
    [a^*a,a^*b^*] &= a^*(aa^*b^* - b^*a^*a) \\
    &= a^*((b^*b+T)b^* - b^*(bb^* + S)) \\
    &= a^*(Tb^*-b^*S).
\end{align}

\noindent Likewise, it follows that
\begin{equation}
    [b^*b, b^*a^*] = b^*(a^*T-Sa^*).
\end{equation}

\noindent Equating the above with the ladder operator relations, it can be seen that
\begin{align}
    [a^*a,a^*b^*] = \lambda a^*b^* &= a^*(Tb^*-b^*S) \\
    [b^*b,b^*a^*] = \lambda' b^*a^* &= b^*(a^*T-Sa^*).
\end{align}

\noindent Since $a^*$ and $b^*$ have trivial kernel, we must have that $Tb^*-b^*S = \lambda b^*$ and similarly $aT-Sa = \lambda'a$. We may use these to prove our result.
\begin{align}
    [b^*a^*, ab] &= b^*(bb^* + S)b - a(a^*a-S)a^* \\
    &= (b^*b)^2 + b^*Sb - (aa^*)^2 + aSa^* \\
    &= (aa^* - T)^2 + b^*Sb - (aa^*)^2 + aSa^*\\
    &= -aa^*T - Taa^* + T^2 + b^*Sb + aSa^* \\
    &= -aa^*T-T(b^*b+T) + T^2 + b^*Sb + aSa^* \\
    &= -a(a^*T-Sa^*) - (Tb^*-b^*S)b \\
    &= -\lambda'aa^* - \lambda b^*b
\end{align}

\noindent Making use of our relations above, it follows that
\begin{equation}
    [b^*a^*,ab] = \mu' aa^*+\nu' = -\lambda'aa^* - \lambda b^*b.
\end{equation}

\noindent Thus
\begin{equation}
    b^*b = -\frac{\lambda'+\mu'}{\lambda}aa^* - \frac{\nu'}{\lambda}
\end{equation}

By repeating these steps with $[a^*b^*,ba]$ and using the above relations $Tb^*-b^*S = \lambda b^*$ and $a^*T-Sa^* = \lambda'a^*$, the claimed affine relationship between $a^*a$ and $bb^*$ holds as well.
\end{proof}

Taking $\lambda=\lambda'$ in the previous theorem gives exactly that $a,b$ form a coupled SUSY as $\alpha = 1 = \beta$ in this case.

It is natural to ask if a coupled SUSY is unique, i.e. given $a$, $\gamma$, $\delta$ if $b$ is unique or if there are many possible choices for $b$. We prove a uniqueness result in the next theorem.

\begin{thm}
Suppose that $\{a,b,\gamma,\delta\}$ and $\{a,c,\gamma,\delta\}$ define coupled SUSYs. Then $b = U (c^* c)^{1/2} $ and $c = V (b^* b)^{1/2}$ for some partial isometries $U,V$.
\end{thm}

\begin{proof}
 Since we have that
 \begin{align}
     bb^* + \gamma &= a^*a = cc^* + \gamma, \\
     b^*b + \delta &= aa^* = c^*c + \delta,
 \end{align}
 It follows that $b^*b = c^*c$ and $bb^* = cc^*$. The equivalences are warranted as the operators are defined on the same subspaces. From the polar decomposition for closed operators \cite{reed1}, the first equality guarantees that $b = U(c^* c)^{1/2}$ for some partial isometry $U$. Switching roles gives $c = V (b^* b)^{1/2}$ for some partial isometry $V$.
\end{proof}

\section{An Energy Ladder Structure for Coupled SUSY}

With ladder operators established for coupled SUSYs, it is natural to inquire about the eigenvalues of the coupled SUSY Hamiltonians $a^*a$, $aa^*$, $b^*b$, and $bb^*$. As in standard SUSY, $a^*a$ and $aa^*$ share the same eigenvalues---up to a possible eigenvalue of $0$. Likewise, $b^*b$ and $bb^*$ share the same eigenvalues---up to a possible eigenvalue of $0$. Moreover, the spectra of $a^*a$ and $bb^*$ are related by a shift of $\gamma$ since $a^*a = bb^* + \gamma$. Thus it is sufficient to study one of $a^*a$ and $aa^*$ to fully understand the eigenvalues of any of the Hamiltonians in a coupled SUSY.

\begin{thm}
If $\{a,b,\gamma,\delta\}$ form an unbroken coupled SUSY, then the eigenvalues of $a^*a$ are given by $m(\delta-\gamma)$ and $m(\delta-\gamma)+\delta$ where $m\in\N_0$.
\end{thm}

\begin{proof}
We first note that if $\psi\in\dom{a}$ is an eigenfunction of $a^*a$, then $a^*b^*\psi\in\dom{a}$; particularly, it is normalizable. To see this, note that $\langle a^*b^*\psi,a^*b^*\psi\rangle = \langle \psi,baa^*b^*\psi \rangle = \langle \psi,b(b^*b+\delta)b^*\psi\rangle < \infty$ since $\psi$ is also an eigenstate of $b^*b$. An analogous result holds for $aa^*$.

Since $a^*b^*$ is a raising operator for $a^*a$ and $0$ is an eigenvalue of $a^*a$, $m(\delta-\gamma)$ is an eigenvalue for $a^*a$. Moreover, $\delta$ is an eigenvalue for $aa^*$ since $0$ is an eigenvalue of $b^*b$. $b^*a^*$ is a raising operator for $aa^*$, so $\delta+m(\delta-\gamma)$ is an eigenvalue for $aa^*$. Since $a^*a$ and $aa^*$ share eigenvalues---up to $0$---$m(\delta-\gamma)+\delta$ is an eigenvalue of $a^*a$.

These are indeed \emph{all} of the eigenvalues for $a^*a$. If $\lambda$ is an eigenvalue of $a^*a$, then so is $\lambda-(\delta-\gamma)$ since $ba$ is a lowering operator for $a^*a$. If an eigenvalue $\lambda$ (corresponding to the eigenfunction $\psi$) existed between $0$ and $\delta-\gamma$, then $\lambda-(\delta-\gamma) < 0$ would be an eigenvalue of $a^*a$ (corresponding to the eigenfunction $ba\psi$) which contradicts the positivity of $a^*a$. Similar reasoning shows that no eigenvalue can exist between $m(\delta-\gamma)$ and $m(\delta-\gamma)+\delta$, which proves the theorem.
\end{proof}

A consequence of this is that there exist no bounded operator representations for a coupled SUSY since the eigenvalues are unbounded. Particularly, no matricial representations exist.

With the ladder structure for $a^*a$ (and $aa^*$) via $a^*b^*$ and $ba$ (and $b^*a^*$ and $ab$), a richer ladder structure exists than the standard SUSY or QMHO structure. We already know that $a$ and $a^*$ transfer between the sectors generated by $a^*a$ and $aa^*$ so we wish to explore the structure that lies beyond this.

In general, there need not be only one state that is annihilated by $a$ (or more generally two states annihilated by $ba$). For instance, it could be that
\begin{equation}
    a = \frac{1}{\sqrt{2}}\left(\begin{array}{cc} \frac{d}{dx}+x & 0 \\ 0 & \frac{d}{dx}+x \end{array}\right),
\end{equation}
which annihilates the states $(\exp(-x^2/2),0)^{\operatorname{T}}$ and $(0,\exp(-x^2/2))^{\operatorname{T}}$. As such we define the following notation.

\begin{defn}
Let $\{a,b,\gamma,\delta\}$ be an unbroken coupled SUSY. Let $\psi_{i,0}$, $i\in I$ for some finite or at most countable index set $I$, be an orthonormal family of (normalized) state vectors annihilated by $a$ and $\phi_{j,0}$, $j\in J$ for some finite or at most countable index set $J$, be an orthonormal family of state vectors annihilated by $ba$ but not annihilated by $a$. Define then $\psi_{i,m} = (a^*b^*)^m\psi_{i,0}/\|(a^*b^*)^m\psi_{i,0}\|$ and $\phi_{j,m} = (a^*b^*)^m \phi_{j,0}/\|(a^*b^*)^m \phi_{j,0}\|$.
\end{defn}

In the case of the harmonic oscillator, the coupled SUSY collapses because $b=a$ and the kernel of $a$ is spanned by
the Gaussian $h_0$, while a vector annihilated by $ba = a^2$ but not by $a$ is necessarily a multiple of the first excited state $h_1$.
Note that in general the states annihilated by $ba$ but not by $a$ are in one-to-one correspondence with the states annihilated by $b$ via the usual SUSY ladder structure.

%

\begin{thm}
Let $\{a,b,\gamma,\delta\}$ define an unbroken coupled SUSY. For any $i\in I$, $j\in J$, and $m\in\N_0$, $\psi_{i,m}$ and $\phi_{j,m}$ are eigenfunctions of $a^*a$ and any normalized eigenfunction of $a^*a$ is of the form $\sum_i \lambda_i \psi_{i,m}$ or $\sum_j \eta_j \phi_{j,m}$ for some finite collection $\lambda_i,\eta_j\in\C$.
\end{thm}

\begin{proof}
It is clear by earlier arguments that $\psi_{i,m},\phi_{j,m}$ are eigenfunctions of $a^*a$. Conversely, every eigenfunction is a linear combination of the $\psi_{i,m}$ or $\phi_{j,m}$ for fixed $m$. To prove this, we proceed by induction on the eigenvalue of a given eigenfunction $\zeta$.

Without loss of generality, we consider the case that $a^*a\zeta = m(\delta-\gamma)\zeta$. The case of $a^*a\zeta = (m(\delta-\gamma)+\delta)\zeta$ proceeds similarly. For $m=0$ this is trivial since $\ker a = \operatorname{span}\{\psi_{i,0}:i\in I\}$ and $\ker a^* = \{0\}$. Assume that $a^*a\zeta = m(\delta-\gamma)\zeta$ implies that $\zeta = \sum_i \lambda_i \psi_{i,m}$ for some finite collection $\lambda_i\in\C$. Suppose then that $a^*a\wt{\zeta} = (m+1)(\delta-\gamma)\wt\zeta$. We wish to show that $\wt\zeta = \sum_i\wt{\lambda_i}\psi_{i,m+1}$ for some $\wt{\lambda_i}\in\C$.

Applying $ba$ to $\wt\zeta$ yields a state with eigenvalue $m(\delta-\gamma)$ and so $ba\zeta = \sum_i\lambda_i\psi_{i,m}$ for some $\lambda_i\in\C$ by the inductive hypothesis. Applying then $a^*b^*$ we have that
\begin{align}
    a^*b^*\sum \lambda_i \psi_{i,m} &= a^*b^*ba\wt\zeta \\
    &= a^*(aa^*-\delta)a\wt\zeta \\
    &= ((a^*a)^2-\delta a^*a)\wt\zeta \\
    &= ((m+1)^2(\delta-\gamma)^2-(m+1)\delta(\delta-\gamma))\wt\zeta.
\end{align}

\noindent Since $\gamma \le 0$ and $m\ge 0$, $(\delta-\gamma)(m+1)-\delta$ is never zero. Thus $\wt\zeta = \sum_i \wt{\lambda_i}\psi_{i,m+1}$ for some $\wt{\lambda_i}$ as claimed.
\end{proof}

As noted above, there is a correspondence between eigenfunctions of $a^*a$ and $aa^*$ via $a$ and $a^*$ in typical SUSY fashion. Thus we can write the eigenfunctions of $aa^*$ as $\wt{\psi}_{i,m} = a\psi_{i,m}/\|a\psi_{i,m}\|$ (where $m\neq 0$ since $a$ annihilates $\psi_{i,0}$) and $\wt{\phi}_{j,m} = a\phi_{j,m}/\|a\phi_{j,m}\|$. As noted above, the states $\wt{\phi}_{j,0}$ are annihilated by $b$. The following figures summarize the coupled SUSY ladder structure.

\begin{figure}[h!]
\centering
\begin{tikzpicture}
\node (a) at (1.5,0) {};
\node (b) at (3.5,0) {};
\draw (a) to node [above] {$a^*a$} (b);

\node (c) at (7.5,0) {};
\node (d) at (9.5,0) {};
\draw (c) to node [above] {$aa^*$} (d);

\draw (1,-5)--(4,-5);
\draw (1,-4)--(4,-4);
\draw (1,-3.5)--(4,-3.5);
\draw (1,-2.5)--(4,-2.5);
\draw (1,-2)--(4,-2);
\draw (1,-1)--(4,-1);
\draw (1,-0.5)--(4,-0.5);

\draw (7,-4)--(10,-4);
\draw (7,-3.5)--(10,-3.5);
\draw (7,-2.5)--(10,-2.5);
\draw (7,-2)--(10,-2);
\draw (7,-1)--(10,-1);
\draw (7,-0.5)--(10,-0.5);

\node (1a) at (4,-4) {};
\node (1b) at (7,-4) {};

\node (2a) at (4,-2.5) {};
\node (2b) at (7,-1) {};

\node at (0.5,-5) {$\psi_{i,0}$};
\node at (0.5,-4) {$\phi_{j,0}$};
\node at (0.5,-3.5) {$\psi_{i,1}$};
\node at (0.5,-2.5) {$\phi_{j,1}$};
\node at (0.5,-2) {$\psi_{i,2}$};
\node at (0.5,-1) {$\phi_{j,2}$};
\node at (0.5,-0.5) {$\psi_{i,3}$};

\node at (10.5,-4) {$\wt\phi_{j,0}$};
\node at (10.5,-3.5) {$\wt\psi_{i,1}$};
\node at (10.5,-2.5) {$\wt\phi_{j,1}$};
\node at (10.5,-2) {$\wt\psi_{i,2}$};
\node at (10.5,-1) {$\wt\phi_{j,2}$};
\node at (10.5,-0.5) {$\wt\psi_{i,3}$};

\draw[->] (1a) to [bend right=30] node[below] {$a$} (1b);
\draw[->] (1b) to [bend right=30] node[below] {$a^*$} (1a);

\draw[->] (2a) to [bend right=30] node[below] {$b^*$} (2b);
\draw[->] (2b) to [bend right=30] node[below] {$b$} (2a);
\end{tikzpicture}
\caption{The actions of $a$, $b$, $a^*$, and $b^*$ in a coupled SUSY.}
\vspace{.5em}
\begin{tikzpicture}
\node (a) at (1.5,0) {};
\node (b) at (3.5,0) {};
\draw (a) to node [above] {$a^*a$} (b);

\node (c) at (7.5,0) {};
\node (d) at (9.5,0) {};
\draw (c) to node [above] {$aa^*$} (d);

\draw (1,-2)--(4,-2);
\draw (1,-1)--(4,-1);
\draw (1,-0.5)--(4,-0.5);

\draw (7,-1)--(10,-1);
\draw (7,-0.5)--(10,-0.5);

\node (1a) at (4,-0.5) {};
\node (1b) at (7,-0.5) {};

\node (2a) at (4,-2) {};
\node (2b) at (7,-0.5) {};

\node (3a) at (0.1,-2) {};
\node (3b) at (0.1,-0.5) {};

\node at (0.5,-2) {$\psi_{i,0}$};
\node at (0.5,-1) {$\phi_{j,0}$};
\node at (0.5,-0.5) {$\psi_{i,1}$};

\node at (10.5,-1) {$\wt\phi_{j,0}$};
\node at (10.5,-0.5) {$\wt\psi_{i,1}$};

\draw[->] (1b) to [bend right=30] node[below] {$a^*$} (1a);

\draw[->] (2a) to [bend right=30] node[below] {$b^*$} (2b);

\draw[->] (3a) to [bend left=30] node[left] {$a^*b^*$} (3b);
\end{tikzpicture}
\caption{The raising operator structure for the first sector in a coupled SUSY.}
\vspace{.5em}
\begin{tikzpicture}
\node (a) at (1.5,0) {};
\node (b) at (3.5,0) {};
\draw (a) to node [above] {$a^*a$} (b);

\node (c) at (7.5,0) {};
\node (d) at (9.5,0) {};
\draw (c) to node [above] {$aa^*$} (d);

\draw (1,-2)--(4,-2);
\draw (1,-1)--(4,-1);
\draw (1,-0.5)--(4,-0.5);

\draw (7,-1)--(10,-1);
\draw (7,-0.5)--(10,-0.5);

\node (1a) at (4,-0.5) {};
\node (1b) at (7,-0.5) {};

\node (2a) at (4,-2) {};
\node (2b) at (7,-0.5) {};

\node (3a) at (0.1,-2) {};
\node (3b) at (0.1,-0.5) {};

\node at (0.5,-2) {$\psi_{i,0}$};
\node at (0.5,-1) {$\phi_{j,0}$};
\node at (0.5,-0.5) {$\psi_{i,1}$};

\node at (10.5,-1) {$\wt\phi_{j,0}$};
\node at (10.5,-0.5) {$\wt\psi_{i,1}$};

\draw[<-] (1b) to [bend right=30] node[below] {$a$} (1a);

\draw[<-] (2a) to [bend right=30] node[below] {$b$} (2b);

\draw[<-] (3a) to [bend left=30] node[left] {$ba$} (3b);
\end{tikzpicture}
\caption{The lowering operator structure for the first sector in a coupled SUSY.}
\end{figure}

\begin{exmp}
 Returning to family in Example 1, a simple inductive proof shows that, for a fixed $n\in\N$, the functions
 \begin{align}
  \psi_{2m}^{(n)}(x) &= e^{\frac{x^{2n}}{2n}} \left(\frac{d}{dx}\frac{1}{x^{2n-2}}\frac{d}{dx}\right)^m e^{-\frac{x^{2n}}{n}} \\
  \psi_{2m+1}^{(n)}(x) &= e^{\frac{x^{2n}}{2n}} \left(\frac{d}{dx}\frac{1}{x^{2n-2}}\frac{d}{dx}\right)^m\left(2x^{2n-1} e^{-\frac{x^{2n}}{n}}\right).
 \end{align}

 \noindent are an orthogonal collection of eigenfunctions of $a_n^*a_n$ satisfying $\psi_{m+2}^{(n)} = \lambda_{m,n} a_n^*b_n \psi_m^{(n)}$ for some $\lambda_{m,n}$. It is not hard to see that these are polynomials multiplying $\exp(-x^{2n}/(2n))$ and analysis similar to that in \cite{williams} shows that they form a basis for $L^2(\R)$. In the case of $n=1$, these become the usual Hermite functions up to normalization. Similar relations hold for the eigenfunctions of $a_n a_n^*$.
\end{exmp}

\section{The Relationship Between Coupled SUSY and Other Oscillator Systems}

Traditionally the QMHO is associated to the 1D Heisenberg-Weyl Lie algebra as this is  the Lie algebra which corresponds to the canonical commutation relations which is reflected in the algebra generated by the ladder operators. This is not the only Lie algebra which may be associated to the QMHO. There are two other treatments of the QMHO: Schwinger's ``spinification" of the two-particle QMHO and the $\mathfrak{su}(1,1)$ treatment of the QMHO. Coupled SUSY is to some degree a unification of the two treatments.

In Schwinger's ``spinification" of the QMHO \cite{schwinger}, one considers two independent oscillators and defines the operators
\begin{equation}
    \mathcal{Q}_{\nu} = \sum_{i,j=1}^2\frac{1}{2} a_i^* (\sigma_{\nu})^{ij} a_j,
\end{equation}
where for $\nu = 1,2,3$, $\sigma_{\nu}$ is the $\nu$th Pauli matrix. Explicitly, $\mathcal{Q}_1 = \frac{1}{2} (a_1^*a_2+a_1a_2^*)$, $\mathcal{Q}_2 = -\frac{i}{2}(a_1^*a_2-a_1 a_2^*)$, and $\mathcal{Q}_3 = \frac{1}{2}(a_1^*a_1 - a_2^*a_2)$. The operators $\mathcal{Q}_{\nu}$ form the $\mathfrak{su}(2)$ Lie algebra, for instance
\begin{align}
    [\mathcal{Q}_1, \mathcal{Q}_2] &= -\frac{i}{4} [a_1^* a_2 + a_1 a_2^*, a_1^* a_2 - a_1 a_2^*] \\
    &= \frac{i}{2} [a_1^* a_2, a_1 a_2^*] \\
    &= \frac{i}{2} (a_1^*a_1 a_2 a_2^* - a_1 a_1^* a_2^* a_2) \\
    &= \frac{i}{2}(a_1^* a_1 - a_2^* a_2) \\
    &= i \mathcal{Q}_3.
\end{align}

The finite dimensional representations in this system have a fixed quantum number which corresponds to the energy difference between the two oscillators and the individual spin states for a fixed energy correspond to the different configurations within each energy level. The Lie algebra $\mathfrak{su}(2)$ is equipped with its own ladder operators. In this case, defining $\mathcal{Q}_{\pm} = \mathcal{Q}_1 \pm i\mathcal{Q}_2$, we have
\begin{equation}
    [\mathcal{Q}_3, \mathcal{Q}_{\pm}] = \pm \mathcal{Q}_{\pm}.
\end{equation}

\noindent A simple computation shows that $\mathcal{Q}_+ = a_1^*a_2$ and $\mathcal{Q}_- = a_2^*a_1$. These are analogous to the 
(quadratic) ladder operators for a coupled SUSY.

Similar to Schwinger's $\mathfrak{su}(2)$ ``spinification" of the QMHO is an $\mathfrak{su}(1,1)$ representation of the QMHO \cite{perelomov}. If $a,a^*$ represent the usual QMHO ladder operators, then letting $\mathcal{K}_0 = a^*a+\frac{1}{2}$, $\mathcal{K}_+ = (a^*)^2$, and $\mathcal{K}_- = a^2$, we have that
\begin{equation}
    [\mathcal{K}_0, \mathcal{K}_{\pm}] = \pm \mathcal{K}_{\pm}, \quad [\mathcal{K}_+, \mathcal{K}_-] = -2\mathcal{K}_0,
\end{equation}

\noindent which is exactly the $\mathfrak{su}(1,1)$ Lie algebra.

Coupled SUSY has elements of both of the above approaches to the QMHO. The ladder operators in all three cases are second order; in coupled SUSY, the ladder operators are a combination of the operators coming from two separate Hamiltonians; and coupled SUSY retains the $\mathfrak{su}(1,1)$ structure implicitly built into the QMHO.

\section{Coherent States for Coupled SUSY Systems}

In traditional SUSY, coherent states have been developed, however the coherent states do not mimic those of the QMHO even though they do exploit the SUSY structure \cite{bagarello,kouri2,david}. Particularly, QMHO coherent states are eigenfunctions of the lowering operator; SUSY coherent states do not enjoy this property or a property similar to it. Part of the reason for this is the lack of a true ladder structure in SUSY. As seen above, coupled SUSY has a rich ladder structure that has elements of both SUSY and the QMHO. This suggests that the coherent states for a coupled SUSY should behave somewhat analogously to that of the QMHO while retaining a SUSY flavor.

In general, there are several kinds of coherent states, depending on the context in which one is interested. Coherent states can be seen to be uncertainty minimizers, eigenstates of a lowering operator, specific infinite series of basis functions, or generalized displacements of a cyclic vector. In the case of the QMHO, these are all the same, however in general this is not the case. For $\mathfrak{su}(1,1)$, the various forms of coherent states have been studied extensively \cite{brif,gerry,perelomov}. We elect to use the displacement operator definition as it uses the Lie algebraic properties of $\mathfrak{su}(1,1)$ and ties well into the coupled SUSY formalism as there are natural cyclic vectors in $\psi_{i,0}$, $\phi_{j,0}$, $\wt\psi_{i,1}$, and $\wt\phi_{j,0}$.

If one has $\mathfrak{su}(1,1)$ operators $\mathcal{K}_0, \mathcal{K}_{\pm}$ where
\begin{equation}
    [\mathcal{K}_0,\mathcal{K}_{\pm}] = \pm \mathcal{K}_{\pm}, \quad [\mathcal{K}_+,\mathcal{K}_-] = -2\mathcal{K}_0, \label{eq:su(1,1)}
\end{equation}

\noindent then the operator $\mathcal{D}(z) = \exp(z\mathcal{K}_+-\bar{z}\mathcal{K}_-)$ defines an $\mathfrak{su}(1,1)$ displacement operator \cite[p. 74]{perelomov}, where $|z|<1$. Suppose that $\mathcal{K}_0\psi_m = (m+k)\psi_m$ for a basis of states $\psi_m$, where $\psi_{m+1} = \mathcal{K}_+\psi_m/\|\mathcal{K}_+\psi_m\|$. The coherent state generated by $\mathcal{D}(z)$ can be written as
\begin{equation}
    |z;k\rangle = \mathcal{D}(z)\psi_0 = (1-|z|^2)^k \sum_{m=0}^{\infty} \left(\frac{\Gamma(m+2k)}{m!\Gamma(2k)}\right)^{\frac{1}{2}} z^m \psi_m.
\end{equation}

Suppose that $\{a,b,\gamma,\delta\}$ defines a coupled SUSY. Define then the following operators
\begin{align}
    \mathcal{K}_0 = \frac{1}{\delta-\gamma}\left(a^*a-\frac{\gamma}{2}\right), \quad \mathcal{K}_+ = \frac{1}{\delta-\gamma} a^*b^*, \quad \mathcal{K}_- = \frac{1}{\delta-\gamma} ba, \\
    \wt{\mathcal{K}_0} = \frac{1}{\delta-\gamma}\left(aa^*-\frac{\delta}{2}\right), \quad \wt{\mathcal{K}_+} = \frac{1}{\delta-\gamma} b^*a^*, \quad \wt{\mathcal{K}_-} = \frac{1}{\delta-\gamma} ab.
\end{align}

\noindent A straightforward calculation shows that the relations in \eqref{eq:su(1,1)} hold. Since we have four families of cyclic vectors (indexed by $i$ and $j$), we have four sets of coherent states (indexed by $i$ and $j$). A simple computation shows that
\begin{align}
    \mathcal{K}_0 \psi_{i,0} &= -\frac{\gamma}{2(\delta-\gamma)} \psi_{i,0} \\
    \mathcal{K}_0 \phi_{j,0} &= \left(\frac{\delta}{2(\delta-\gamma)}+\frac{1}{2}\right) \phi_{j,0} \\
    \wt{\mathcal{K}_0} \wt\psi_{i,1} &= \left(-\frac{\gamma}{2(\delta-\gamma)}+\frac{1}{2}\right) \wt\psi_{i,1} \\
    \wt{\mathcal{K}_0} \wt\phi_{j,0} &= \frac{\delta}{2(\delta-\gamma)}\wt\phi_{j,0},
\end{align}

\noindent from which we get the following coherent states:
\begin{align}
    \left|z; -\frac{\gamma}{2(\delta-\gamma)}\right\rangle_i &= (1-|z|^2)^{-\frac{\gamma}{2(\delta-\gamma)}} \sum_{m=0}^{\infty} \left(\frac{\Gamma\left(m-\frac{\gamma}{\delta-\gamma}\right)}{m!\,\Gamma\left(-\frac{\gamma}{\delta-\gamma}\right)}\right)^{\frac{1}{2}} z^m \psi_{i,m} \label{eq:coherent_state_1} \\
    \left|z;\frac{\delta}{2(\delta-\gamma)}+\frac{1}{2}\right\rangle_j &= (1-|z|^2)^{\frac{\delta}{2(\delta-\gamma)}+\frac{1}{2}} \sum_{m=0}^{\infty} \left(\frac{\Gamma\left(m+\frac{\delta}{\delta-\gamma}+1\right)}{m!\, \Gamma\left(\frac{\delta}{\delta-\gamma}+1\right)}\right)^{\frac{1}{2}} z^m \phi_{j,m} \label{eq:coherent_state_2}\\
    \left|z;-\frac{\gamma}{2(\delta-\gamma)}+\frac{1}{2}\right\rangle_i &= (1-|z|^2)^{-\frac{\gamma}{2(\delta-\gamma)}+\frac{1}{2}} \sum_{m=1}^{\infty} \left(\frac{\Gamma\left(m-\frac{\gamma}{\delta-\gamma}+1\right)}{m!\, \Gamma\left(-\frac{\gamma}{\delta-\gamma}+1\right)}\right)^{\frac{1}{2}} z^m\wt\psi_{i,m} \label{eq:coherent_state_3} \\
    \left|z; \frac{\delta}{2(\delta-\gamma)}\right\rangle_j &= (1-|z|^2)^{\frac{\delta}{2(\delta-\gamma)}} \sum_{m=0}^{\infty} \left(\frac{\Gamma\left(m+\frac{\delta}{\delta-\gamma}\right)}{m!\,\Gamma\left(\frac{\delta}{\delta-\gamma}\right)}\right)^{\frac{1}{2}} z^m \wt\phi_{j,m} \label{eq:coherent_state_4}
\end{align}

In the case of the QMHO, the coherent states are eigenstates of the lowering operator. Since coupled SUSY behaves so similarly to the QMHO, it is natural to ask what happens to the coupled SUSY coherent states under an application of the lowering operators. The lowering operators are composed of $a$ and $b$, so we want to investigate the action of $a$ and $b$ on the coherent states. To this end, we have the following lemma.

\begin{lem}
Let $\psi_{i,m}$, $\phi_{j,m}$, $\wt\phi_{i,m}$, and $\wt\phi_{j,m}$ be as above. Then
\begin{align}
    a\psi_{i,m} &= \sqrt{m(\delta-\gamma)} \wt{\psi}_{i,m} \label{eq:lowering_1} \\
    a\phi_{j,m} &= \sqrt{(\delta-\gamma)\left(m+\frac{\delta}{\delta-\gamma}\right)} \wt{\phi}_{j,m} \\
    b\wt{\psi}_{i,m} &= \sqrt{(\delta-\gamma)\left(m-\frac{\delta}{\delta-\gamma}\right)} \psi_{i,m-1} \\
    b\wt{\phi}_{j,m} &= \sqrt{m(\delta-\gamma)} \phi_{j,m-1} \label{eq:lowering_4}
\end{align}
\end{lem}

\begin{proof}
We only prove the first as the others follow in exactly the same manner. We assume without loss of generality that the proportionality is pure real as a global phase does not change the underlying mathematics. By definition, $a\psi_{i,m} = \lambda \wt\psi_{i,m}$, so we need only to solve for $\lambda$.
\begin{align}
    |\lambda|^2 &= \langle a\psi_{i,m}, a\psi_{i,m}\rangle \\
    &= \langle a^*a \psi_{i,m},\psi_{i,m}\rangle \\
    &= m(\delta-\gamma).
\end{align}

\noindent The last equality follows since $\psi_{i,m}$ is an eigenfunction of $a^*a$ with eigenvalue $m(\delta-\gamma)$. Thus $\lambda = \sqrt{m(\delta-\gamma)}$ as desired.
\end{proof}

Since our ladder operators can only relate $\psi_{i,m}$ with $\wt\psi_{i,m'}$ and $\phi_{j,m}$ with $\wt\phi_{j,m'}$ (and vice versa), it is clear that we can only relate \eqref{eq:coherent_state_1} with \eqref{eq:coherent_state_3} and \eqref{eq:coherent_state_2} with \eqref{eq:coherent_state_4} (and vice versa). Making use of relations \eqref{eq:lowering_1}-\eqref{eq:lowering_4}, we have that
\begin{align}
    a\left|z;-\frac{\gamma}{2(\delta-\gamma)}\right\rangle_i &= \sqrt{-\gamma} \frac{z}{\sqrt{1-|z|^2}} \left|z;-\frac{\gamma}{2(\delta-\gamma)}+\frac{1}{2}\right\rangle_i, \\
    b\left|z;\frac{\delta}{2(\delta-\gamma)}\right\rangle_j &= \sqrt{\delta} z\sqrt{1-|z|^2} \left|z;-\frac{\delta}{2(\delta-\gamma)}+\frac{1}{2}\right\rangle_j.
\end{align}

If we were to apply $b$ to the first or $a$ to the second, we would not retain a multiple of the original state as the relations \eqref{eq:lowering_1}-\eqref{eq:lowering_4} indicate. Thus while the coherent states are not eigenstates of the lowering operators $ba$ or $ab$, applying half of one of the lowering operators can convert one coherent state into another---up to a multiplicative factor (since the operators are not unitary). Resolutions of the identity hold for these coherent states which come directly from the resolution of the identity for $\mathfrak{su}(1,1)$ coherent states as these are simply $\mathfrak{su}(1,1)$ coherent states, e.g. taking $\mu_i(z) = -\frac{\delta}{\pi(\delta-\gamma)} \frac{d^2 z}{(1-|z|^2)^2}$, we have that
\begin{equation}
 \int_{D(0,1)} |z\rangle\langle z| \,d\mu_i(z) = I
\end{equation}

\noindent where the identity is interpreted as being the identity on the closed subspace generated by the $\psi_{i,m}$. This can be easily established by considering inner products with $\psi_{i,m}$. As far as the authors are aware, this is a new structure in SUSY and QMHO coherent states.

\section{Uncertainty Principles for Coupled SUSY Systems}

The canonical uncertainty principle in quantum mechanics is the Heisenberg uncertainty principle which is an uncertainty principle between the position operator $x$ and the momentum operator $p$. The Heisenberg uncertainty principle says that, in natural units, the standard deviation in 
position and momentum is bounded below by
$
    \sigma_x \sigma_p \ge \frac{1}{2}.
$

The minimizer of the uncertainty principle is the Gaussian (and translations and modulations thereof). This is easily proved via Cauchy-Schwarz techniques \cite{grochenig}. Since $x$ and $p$ can be written as linear combinations of the QMHO ladder operators (i.e. $x = \frac{1}{\sqrt{2}}(a+a^*)$ and $p = \frac{i}{\sqrt{2}}(a-a^*)$),  we expect to realize uncertainty principles for coupled SUSYs in a similar way via their ladder operators.

\begin{defn}
Let $\{a,b,\gamma,\delta\}$ define a coupled SUSY. We define the following analogues of the traditional position and momentum operators for the separate sectors:
\begin{align}
    \mathcal{L} &= - \frac{1}{2}(a^*b^*+ba), \quad \mathcal{A} = \frac{i}{2}(a^*b^* - ba), \\
    \widetilde{\mathcal{L}} &= -\frac{1}{2}(b^*a^* + ab), \quad \widetilde{\mathcal{A}} = \frac{i}{2}(b^*a^* - ab).
\end{align}
\end{defn}

In the case of the family of examples in Example 1, $\mathcal{L} = -\frac{1}{2} \frac{d}{dx}\frac{1}{x^{2n-2}}\frac{d}{dx} - \frac{1}{2} x^{2n}$ resembles a Lagrangian and is exactly a Lagrangian operator $-\frac{1}{2} \frac{d^2}{dx^2} -\frac{1}{2} x^2$ when $n=1$, whereas $\mathcal{A} = \frac{1}{2}\{x,p\}$ is precisely the dilation operator regardless of $n$. This motivates our notation.

\begin{thm}
Let $\{a,b,\gamma,\delta\}$ be an unbroken coupled SUSY, $\mathcal{L}$ and $\mathcal{A}$ be as above, and, also as above, $\ker a = \{\psi_{i,0}:i\in I\}$ for some index set $I$. An uncertainty principle holds for $\mathcal{L}$ and $\mathcal{A}$ and the minimizers are the states $\psi_{i,0}$.
\end{thm}

\begin{proof}
 Let $\psi$ be a normalized wavefunction. Note that Robertson's uncertainty relation gives us that
 \begin{align}
     (\sigma_{\mathcal{L}}\sigma_{\mathcal{A}})_{\psi} & \ge \frac{1}{2}|\langle\psi|[\mathcal{L},\mathcal{A}]|\psi\rangle| \\
     &= \frac{1}{4}|\langle\psi|[a^*b^*,ba]-[ba, a^*b^*]|\psi\rangle| \\
     &= \frac{\delta-\gamma}{4}|2\langle\psi|a^*a|\psi\rangle - \gamma|.
 \end{align}
 
 \noindent Since $\gamma<0$, $a$ has annihilating states $\psi_{i,0}$, and $a^*a$ is self-adjoint, the lower bound given by these states is
 \begin{equation}
     \sigma_{\mathcal{L}}\sigma_{\mathcal{A}} \ge \frac{1}{4}(\delta-\gamma)|\gamma|.
 \end{equation}
 
 This does not guarantee that this lower bound is indeed attained. For the states $\psi_{i,0}$, we have
 \begin{align}
     \langle \mathcal{L}\rangle &= \langle \phi_{i,0}| \mathcal{L}|\psi_{i,0}\rangle \\
     &= \frac{1}{2}\langle \phi_{i,0} |a^*b^*+ba|\psi_{i,0}\rangle \\
     &= 0.
 \end{align}
 
 \noindent Similarly, $\langle \mathcal{A}\rangle = 0$. Evaluating $\langle \mathcal{L}^2\rangle$, making use of the fact that $a$ annihilates the states $\psi_{i,0}$, and employing the coupled SUSY structure, it follows that
 \begin{align}
     \langle\mathcal{L}^2\rangle &= \frac{1}{4}\langle\psi_{i,0}|(ba + a^*b^*)^2|\psi_{i,0}\rangle \\
     &= \frac{1}{4}\langle\psi_{i,0}|baa^*b^*|\psi_{i,0}\rangle \\
     &= \frac{1}{4}\langle\psi_{i,0}|b(b^*b+\delta)b^*|\psi_{i,0}\rangle \\
     &= \frac{1}{4}\langle\psi_{i,0}|(a^*a-\gamma)^2+\delta(a^*a-\gamma)|\psi_{i,0}\rangle \\
     &= \frac{1}{4}(\delta-\gamma)|\gamma|.
 \end{align}
 
 \noindent An identical result holds for $\langle\mathcal{A}^2\rangle$, thus we have that
 \begin{align}
     \sigma_{\mathcal{L}}\sigma_{\mathcal{A}} &= \sqrt{\langle\mathcal{L}^2\rangle - \langle\mathcal{L}\rangle^2}\sqrt{\langle\mathcal{A}^2\rangle - \langle\mathcal{A}\rangle^2} \\
     &= \frac{1}{4}(\delta-\gamma)|\gamma|.
 \end{align}
 
 The states $\psi_{i,0}$ are indeed the \emph{only} minimizers. If $a\psi \neq 0$, then $\langle\psi|a^*a|\psi\rangle > 0$. Since $\gamma \le 0$, it follows that the uncertainty product is strictly greater than $\frac{1}{4}(\delta-\gamma)|\gamma|$.
\end{proof}

\begin{thm}
 Let $\{a,b,\gamma,\delta\}$ define an unbroken coupled SUSY, $\widetilde{\mathcal{L}}$ and $\widetilde{\mathcal{A}}$ be as above, and, also as above, $\ker b^* = \{\wt\phi_{j,0}:j\in J\}$ for some index set $J$. An uncertainty principle holds for $\widetilde{\mathcal{L}}$ and $\widetilde{\mathcal{A}}$ and the minimizers are the states $\wt\phi_{j,0}$.
\end{thm}

The proof of this theorem is nearly identical to that of the previous so we omit it, however the uncertainty principle is now given by
\begin{equation}
    \sigma_{\widetilde{\mathcal{L}}}\sigma_{\widetilde{\mathcal{A}}} \ge \frac{1}{4}(\delta-\gamma)\delta.
\end{equation}

Often in SUSY one treats the first and second sectors simultaneously in a matrix formulation by defining the operators $\mathcal H$, $\mathcal Q$, and $\mathcal Q^*$ which act on the direct sum of the two sectors as follows:
\begin{align}
    \mathcal H = \left(\begin{array}{cc} \mathcal{H}_1 & 0 \\ 0 & \mathcal{H}_2\end{array}\right), \quad \mathcal Q = \left(\begin{array}{cc} 0 & 0 \\ \mathcal{Q}_1 & 0 \end{array}\right), \quad \mathcal Q^* = \left(\begin{array}{cc} 0 & \mathcal{Q}_1^* \\ 0 & 0 \end{array}\right).
\end{align}

The logic being that the joint Hamiltonian should act on the two subspaces separately by their own Hamiltonians and is therefore diagonal, and the joint charge operators should be off-diagonal because $\mathcal{Q}_1$ and $\mathcal{Q}_1^*$ transfer between the two sectors.

This allows us to define a tertiary set of first order position and momentum operators. Previously, the analysis was relegated to second order position and momentum operators because the operators $\mathcal{L}$, $\mathcal{A}$, $\wt{\mathcal{L}}$, and $\wt{\mathcal{A}}$ acted within a sector; however by combining the two sectors into one framework, we allow ourselves the ability to drop down to first order operators, analogous to the usual QMHO case.

\begin{defn}
Again, let $\{a,b,\gamma,\delta\}$ define a coupled SUSY. We define the operators $\mathcal{X}$ and $\mathcal{P}$ on the direct sum of the two sectors as follows:
\begin{equation}
    \mathcal{X} = \frac{1}{\sqrt{2}} \left(\begin{array}{cc} 0 & a^*+b \\ a+b^* & 0 \end{array}\right), \quad \mathcal{P} = -\frac{i}{\sqrt{2}} \left(\begin{array}{cc} 0 & a^*-b \\ -a+b^* & 0 \end{array}\right).
\end{equation}
\end{defn}

For the infinite family of operators in Example 1, $a_n^*+b_n$ yields $\sqrt{2} x^n$ and $a_n^*-b_n$ yields $\frac{\sqrt{2}}{x^{n-1}}\frac{d}{dx}$. Hence $a_n^*+b_n$ extracts the coordinate-like object corresponding to the coupled SUSY, whereas $a_n^*-b_n$ extracts the derivative-like object corresponding to the coupled SUSY. Therefore, 
$\mathcal{X}$ plays the role of a generalized notion of position and $\mathcal{P}$  plays the role of a generalized notion of momentum \cite{kouri2}.

\begin{thm}
Let $\{a,b,\gamma,\delta\}$ define an unbroken coupled SUSY and $\mathcal{X}$ and $\mathcal{P}$ be as above, then the following uncertainty principle holds for $\mathcal{X}$ and $\mathcal{P}$:
\begin{equation}
    \sigma_{\mathcal{X}}\sigma_{\mathcal{P}} \ge \frac{1}{2}\min\{|\gamma|,\delta\}.
\end{equation}
\end{thm}

\begin{proof}
Let $\Psi = (\psi_1, \psi_2)^{\operatorname{T}}$ be the state in which we are evaluating the expectation, then $1 = \|\psi_1\|^2 + \|\psi_2\|^2$. Again making use of Robertson's inequality, we have that
\begin{align}
    (\sigma_{\mathcal{X}}\sigma_{\mathcal{P}})_{\Psi} & \ge \frac{1}{2} |\langle\Psi| [\mathcal{X}, \mathcal{P}]|\Psi\rangle| \\
    &= \frac{1}{4} \left|\left\langle\Psi\left|\left[\left(\begin{array}{cc} 0 & a^*+b \\ a+b^* & 0 \end{array}\right),\left(\begin{array}{cc} 0 & a^*-b \\ -a+b^* & 0 \end{array}\right)\right]\right|\Psi\right\rangle\right| \\
    &= \frac{1}{4} \left|\left\langle\Psi\left|\left(\begin{array}{cc} -2\gamma & 0 \\ 0 & 2\delta \end{array}\right)\right|\Psi\right\rangle\right| \\
    &= \frac{1}{2} (|\gamma|\|\psi_1\|^2 + \delta\|\psi_2\|^2).
\end{align}

\noindent Since $\|\psi_2\|^2 = 1-\|\psi_1\|^2$, the above is a convex combination of $|\gamma|$ and $\delta$, so indeed we have that
\begin{equation}
    \sigma_{\mathcal{X}}\sigma_{\mathcal{P}} \ge \frac{1}{2}\min\{|\gamma|,\delta\}.
\end{equation}

We now show that the value of $\frac{1}{2}|\gamma|$ is attainable. Let $\Psi = (\psi_{i,0},0)^{\operatorname{T}}$, where $\psi_{i,0}$ is as above. The case of $\frac{1}{2}\delta$ proceeds similarly by taking $\Psi = (0,\wt\phi_{j,0})^{\operatorname{T}}$, where $\wt\phi_{j,0}$ is also as above. For this choice of $\Psi$, it follows that
\begin{align}
    \langle\Psi|\mathcal{X}|\Psi\rangle &= \frac{1}{\sqrt{2}}(\langle\psi_{i,0} |a^* + b|0\rangle + \langle 0| a+b^*|\psi_{i,0}\rangle) = 0, \\
    \langle\Psi|\mathcal{P}|\Psi\rangle &= -\frac{i}{\sqrt{2}}(\langle\psi_{i,0} |a^*-b|0\rangle + \langle 0|-a+b^*|\psi_{i,0}\rangle) = 0.
\end{align}

\noindent Computing $\mathcal{X}^2$ and $\mathcal{P}^2$ yields
\begin{align}
    \mathcal{X}^2 &= \frac{1}{2}\left(\begin{array}{cc} (a^*+b)(a+b^*) & 0 \\ 0 & (a+b^*)(a^*+b)\end{array}\right) \\
    &= \frac{1}{2}\left(\begin{array}{cc} a^*a + a^*b^* + ba + bb^* & 0 \\ 0 & aa^* + ab + b^*a^* + b^*b\end{array}\right) \\
    \mathcal{P}^2 &= \frac{1}{2}\left(\begin{array}{cc} (a^*-b)(a-b^*) & 0 \\ 0 & (a-b^*)(a^*-b)\end{array}\right) \\
    &= \frac{1}{2} \left(\begin{array}{cc} a^*a - a^*b^* - ba + bb^* & 0 \\ 0 & aa^* - ab - b^*a^* + b^*b \end{array}\right)
\end{align}

Inspecting the diagonal terms, it is clear that this uncertainty principle is quite different from that of $\mathcal{L}$ and $\mathcal{A}$ (and from that of $\widetilde{\mathcal{L}}$ and $\widetilde{\mathcal{A}}$). We are only concerned with the upper left elements since we are considering states of the form $\Psi = (\psi_{i,0}, 0)^{\operatorname{T}}$. Noting that $a$ annihilates $\psi_{i,0}$ and using the coupled SUSY equations, it follows that
\begin{equation}
    \langle \Psi|\mathcal{X}^2|\Psi\rangle = -\frac{1}{2}\gamma = \langle\Psi|\mathcal{P}^2|\Psi\rangle.
\end{equation}

\noindent Since $\gamma \le 0$, the result follows.
\end{proof}

\begin{rem}
The above uncertainty principles agree exactly with the traditional Heisenberg uncertainty principle in the case of the QMHO since $\gamma = -1$ and $\delta = 1$, giving an uncertainty bound of $\frac{1}{2}$ for each with the minimizers being Gaussians. For $n>1$, the uncertainty product $\sigma_{\mathcal{X}}\sigma_{\mathcal{P}}$ has a lower bound of $\frac{1}{2}$, just as in the Heisenberg uncertainty principle with minimizers $\exp(-x^{2n}/2n)$, but the uncertainty products $\sigma_{\mathcal{L}}\sigma_{\mathcal{A}}$ and $\sigma_{\wt{\mathcal{L}}}\sigma_{\wt{\mathcal{A}}}$ are for $n>1$ bounded by a larger constant  since $\delta-\gamma = 2n$.
\end{rem}

\section{Realizing Harmonic Oscillators With Coupled Supersymmetry}

As previously noted, the quantum mechanical harmonic oscillator is a specific instance of a coupled supersymmetry, albeit a somewhat trivial case in which the two coupled SUSY equations are identical. This is not the only manner in which the two are connected. Indeed, a special class of coupled SUSYs may be realized as harmonic oscillator-like systems, i.e. they satisfy the same Lie algebra and by virtue of Stone-von Neumann, may be realized in some way as harmonic oscillators. If one takes $\gamma = -\delta$, then the coupled SUSY equations become
\begin{align}
 a^*a &= bb^* -\delta, \\
 aa^* &= b^*b + \delta.
\end{align}

As in the preceding section, we define matricial operators $\mathcal{A}$ and its adjoint $\mathcal{A}^*$ acting on $\mathfrak{H}_1\times\mathfrak{H}_2$ by
\begin{equation}
 \mathcal{A} = \begin{pmatrix} 0 & b \\ a & 0 \end{pmatrix}, \qquad \mathcal{A}^* = \begin{pmatrix} 0 & a^* \\ b^* & 0 \end{pmatrix}.
\end{equation}

\noindent $\mathcal{A}$ and $\mathcal{A}^*$ are linear combinations of the matricial $\mathcal{X}$ and $\mathcal{P}$ operators from the previous section.

These operators have a simple commutation relation: $[\mathcal{A},\mathcal{A}^*] = \delta I_{\mathfrak{H}_1\times\mathfrak{H}_2}$. It is natural to follow the usual prescription of supersymmetric quantum mechanics by picking the matrix Hamiltonian $\mathcal{H}$ to be
\begin{equation}
 \mathcal{H} = \begin{pmatrix} a^*a & 0 \\ 0 & aa^*\end{pmatrix},
\end{equation}

\noindent however a straightforward calculation shows that appending $\mathcal{H}$ to the Lie algebra generated by $I_{\mathfrak{H}_1\times\mathfrak{H}_2}$, $\mathcal{A}$, and $\mathcal{A}^*$ does not lead to the standard quantum mechanical harmonic oscillator Lie algebra. Define $\mathcal{H}$ instead by
\begin{equation}
 \mathcal{H} = \begin{pmatrix} a^*a & 0 \\ 0 & b^*b \end{pmatrix}.
\end{equation}

\noindent With this choice of $\mathcal{H}$, it follows that $[\mathcal{H},\mathcal{A}] = -\delta \mathcal{A}$ and $[\mathcal{H},\mathcal{A}^*] = \delta \mathcal{A}^*$ and the harmonic oscillator Lie algebra emerges---up to rescaling.

For ease of notation and without much loss of generality, assume that $\dim \ker a = 1=  \dim \ker b$ so that the eigenstates of $a^*a$ can be indexed by $\psi_k$ and $\phi_k$ and similarly for $b^*b$. The operators $\mathcal{X}$ and $\mathcal{P}$ are linear combinations of $\mathcal{A}$ and $\mathcal{A}^*$ and are both self-adjoint, however both are reducible, with invariant subspaces $\overline{\operatorname{span}}(\psi_k,\widetilde{\psi}_{k'})^{\text{T}}$ and $\overline{\operatorname{span}}(\phi_k,\widetilde{\phi}_{k'})^{\text{T}}$, cf. Figure 1. These subspaces are orthogonal to each other and are irreducible. This class of coupled SUSYs restricted to these subspaces gives that the creation and annihilation operators act as the harmonic oscillator creation and annihilation operators by Stone-von Neumann \cite[p. 286]{hall}. Similarly, $\mathcal{H}$ acts as the quantum mechanical harmonic oscillator Hamiltonian on these subspaces.

Thus the class of coupled SUSYs for which $-\gamma = \delta$ can represent a harmonic oscillator with a caveat: the states have two components and there are two ground states. The states $(\psi_0, 0)^{\text{T}}$ and $(0, \widetilde{\phi}_0)^{\text{T}}$ have eigenvalue $0$ under $\mathcal{H}$. These two states generate their respective invariant subspaces, much in the same way that the Gaussian generates the excited states for the harmonic oscillator. Another caveat is that it is not \emph{a priori} obvious that the eigenstates form a basis for $\mathfrak{H}_1\times\mathfrak{H}_2$ unlike in the case of the harmonic oscillator, however one may simply work in the Hilbert space generated by the eigenstates without much loss of generality.

That harmonic oscillators can be realized with coupled SUSYs suggest that coupled SUSY may have applications to general vibrational systems, including molecules, solid state systems, quantum field theory, and elementary particles. Particularly, a quantum field theory built from coupled SUSY in this way is naturally supersymmetric and has a natural correspondence between multiparticle systems and excited states. Coupled SUSYs may be better suited to molecular physics problems involving high order potentials or potentials with steep wells. Note that if instead $-\gamma\neq\delta$, the canonical commutation relations do not quite hold---instead the commutator is a constant diagonal matrix---and the multiparticle interpretation for excited states does not quite hold.

\section{Conclusion and Future Work}

By considering the QMHO in the context of SUSY, the coupled SUSY structure unifying the QMHO and SUSY was developed. Coupled SUSY has many of the desirable properties of both: true ladder operators exist, there are two sectors, and charge operators exist between the sectors. The existence of true ladder operators led to a richer theory for coherent states than what has existed in the past, namely applying a charge operator led to an intertwining in the coherent states structures in both sectors. Moreover, focusing on the case of unbroken coupled SUSY gave some results regarding the spectrum of coupled SUSY Hamiltonians as well as uncertainty principles which generalize the Heisenberg uncertainty principle. A strength of the theory is its background-independence, e.g., the above may apply to any $L^2(\Omega)$ space. $\Omega$ may simply be a locally compact group or even have a manifold structure. Unbroken coupled SUSYs for which $-\gamma =\delta$ can be realized as harmonic oscillators and as such coupled SUSY may be a better model for elementary particles than traditional supersymmetric quantum mechanics. Moreover, coupled SUSY is inherently a multi-particle and multi-dimensional theory and does not suffer some of the pitfalls that traditional SUSY quantum mechanics has for multiple particles or multiple dimensions \cite{kouri4,mandelshtam,kouri5}. It has also been shown that spurious states exist for sector two tensor Hamiltonians for tensorial systems \cite{chia1} In spite of these drawbacks, tensorial multidimensional SUSY quantum mechanics has been used to determine excited-state energies and wavefunctions using adiabatic switching \cite{chia2} in addition to the construction of sector one states via imaginary time propagation \cite{chia3}.

In a future work, the authors intend to explore broken coupled SUSY, the nature of the coupled SUSY algebra, and the spectral theory of the coupled SUSY momentum and position operators. The ladder structure for coupled SUSY appears to have a natural relationship with the group of Lorentz transformations $\operatorname{SO}(2,1)$,
which may be promising for applications to quantum field theory. On the other hand, the perturbed eigenvalue structure in comparison with the 
standard harmonic oscillator may also have implications in solid state theory. More generally, the appearance of Lagrangian and dilation variables is intriguing 
for the formulations of quantum field theories. The appearance of the Lagrangian is also of interest for describing tunnelling processes.

\section{Acknowledgements}

\noindent We thank John R. Klauder for his helpful comments and suggestions in the drafting of this manuscript. We congratulate Professor Michael Baer on his 80th birthday and wish him many more years!

\hfill

\noindent This research was supported in part under Grant E-0608 from the Robert A. Welch Foundation. B.G.B was supported in part by NSF grant DMS-1412524.

\bibliographystyle{plain}
\bibliography{coupled_susy_ref}

\end{document}